\documentclass[letterpaper, 10 pt, conference]{ieeeconf}
\usepackage{hyperref}
\usepackage{amsmath,mathrsfs}
\usepackage{graphicx}
\usepackage{amsfonts}
\usepackage{amssymb}
\usepackage{tikz}
\usepackage{color}

\IEEEoverridecommandlockouts    

\newtheorem{theorem}{Theorem}

\newtheorem{lemma}[theorem]{Lemma}

\begin{document}

\title{Physical Reduced Stochastic Equations for Continuously Monitored Non-Markovian Quantum Systems with a Markovian Embedding}

\author{Hendra I. Nurdin\thanks{H. I. Nurdin is with the School of Electrical Engineering and 
Telecommunications,  UNSW Australia,  Sydney NSW 2052, Australia (\texttt{email: h.nurdin@unsw.edu.au})}}

\maketitle

\begin{abstract}
An effective approach to modeling non-Markovian quantum systems is to embed a principal (quantum) system  of interest into a larger quantum system. A widely employed embedding is one that uses another quantum system, referred to  as the auxiliary system, which is coupled to the principal system, and both the principal and auxiliary can be coupled to quantum white noise processes. The principal and auxiliary together form a quantum Markov system and the quantum white noises act as a bath (environment) for this system. 

Recently it was shown that the conditional  evolution of the principal system in this embedding under continuous monitoring by a travelling quantum probe can be expressed as a system of  coupled stochastic differential equations (SDEs) that involve only operators of the principal system. The reduced conditional  state of the principal only (conditioned on the measurement outcomes) is determined by the ``diagonal" blocks of this coupled systems of SDEs. It is shown here that the ``off-diagonal" blocks can be exactly eliminated up to their initial conditions, leaving a reduced closed system of SDEs  for the diagonal blocks only. Under additional conditions the  off-diagonal initial conditions can be made to vanish. This new closed system of equations, which includes an integration term involving a two-time stochastic kernel, represents the non-Markovian stochastic dynamics of the principal system under continuous-measurement. The system  of equations determine the reduced conditional state of the principal only and may be viewed as a stochastic Nakajima-Zwanzig type of equation for continuously monitored non-Markovian quantum systems. \textcolor{blue}{[This work has been published as IEEE Control Systems Letters 9, pp. 1009-1014 (2025)]}
\end{abstract}

\section{Introduction}

Quantum Markov models are based on the scenario of a quantum system of interest, referred to herein as the {\em principal (quantum) system}, being coupled to one or more quantum white noise processes as its environment  \cite{HP84,GZ04,BvHJ07,WM10}.  They have been ubiquitously employed as accurate models for various physical systems, including  quantum optical, optomechanical and superconducting systems, see, e.g,. \cite{GZ04,WM10,NY17}, as well as a model for quantum noise in quantum computers \cite[Chapter 8]{NC10}. However, they are not suitable for modeling all physical systems, including environments that retain a memory of past states of the principal system, for which non-Markovian models are required. Non-Markovian quantum systems also arise in quantum technologies, for example in systems driven by a quantum field in various types of non-classical non-Gaussian states  \cite{GJN11,GJNC12,GJN13,GJN14}.  

Quantum Markov models also underpin the theory for continuously-monitored quantum systems, common  in quantum optics \cite{BvHJ07,WM10} and superconducting quantum systems \cite{Clerk10}. In this case, the principal system is coupled to a traveling probe that is modeled as a quantum white noise. The evolution of the quantum system under this continuous-measurement is described by a quantum filtering equation \cite{BvHJ07}, also commonly referred to in the physics literature as a stochastic master equation \cite{WM10}. 

One approach to modeling non-Markovian systems is to embed the principal system in a larger quantum system. Often the larger system is taken to be a Markovian quantum system that includes another quantum system, which will be referred to herein as an {\em auxiliary (quantum) system}, and quantum white noises as a bath for the principal and auxiliary. The principal and auxiliary can be coupled to one another and both can be coupled to the quantum white noises. The principal and auxiliary thus form a Markovian quantum system while the auxiliary and the quantum white noises form what is referred to as a compound noise source in \cite{Nurd23}. Examples include taking the auxiliary to be a system of quantum harmonic oscillators and the quantum white noises to be bosonic \cite{Imamoglu94,DBG01,Mascherpa20}. This approach has also been adopted to the case where the bath consists of fermionic quantum white noises, see, e.g., \cite{CAG19}. Another type of Markovian embedding  takes as the compound noise the continuous-mode output of a quantum input-output system \cite{GZ04,CKS17} that is driven by quantum white noise fields. The principal system is another quantum input-output system that is coupled to this compound noise via a cascade connection in which the output from the compound noise drives the input of the principal system \cite{GJN11,GJNC12,GJN13,GJN14}. 

The quantum filtering equation has a dual use for numerically computing the unconditional state of a quantum Markov system as the solution to the  Gorini-Kossakowski-Sudarshan-Lindblad (GKSL) quantum master equation, via Monte Carlo simulations; see \cite{HC93} and the references therein. It is based on simulating  stochastic pure states that solve a stochastic Schr\"{o}dinger equation (SSE) driven by a Wiener or Poisson process.  The  noise driving the SSE has a physical interpretation of being the back-action noise due to continuous measurement. A non-Markovian version of the SSE called non-Markovian quantum state diffusion (NMQSD)  has also been proposed \cite{DGS98}. It is is based on stochastic pure states that solve a non-Markovian SDE driven by a complex Gaussian white noise. The NMQSD can be used to numerically estimate, also via Monte Carlo simulations, the solution of a non-Markovian quantum master equation that gives the reduced (unconditional) state of a class of non-Markovian quantum systems.  However, unlike the Markovian SSEs, the quantum state diffusion equation in NMQSD is not known to be associated to some continuous-measurement process on the non-Markovian quantum system, thus lacking a physical interpretation \cite{KS12,GDA25}. 

In \cite{Nurd23} it was shown that the conditional  evolution of the principal system in a Markovian embedding under continuous monitoring by a travelling quantum probe can be expressed as a system of  coupled stochastic differential equations (SDEs) that involve only operators on the principal system. Likewise, the unconditional evolution of the principal system is given by a system of coupled ordinary differential equations (ODEs) involving only operators on the principal system.
The reduced conditional  state  (conditioned on the measurement outcomes) and unconditional state of the principal only are determined only by the ``diagonal" blocks of this coupled system  of SDEs and ODEs, respectively.
In this paper it is shown that ``off-diagonal" blocks of the system of SDEs can be exactly eliminated up to their initial conditions, leaving only a closed system of SDEs  for the diagonal blocks. Under additional conditions the  off-diagonal initial conditions can be made to vanish. This reduced closed system of equations, which includes an integration term with a two-time stochastic kernel, represents the non-Markovian stochastic dynamics of the principal system under continuous measurement. These equations determine  the reduced conditional state of the principal only and can be viewed as  a stochastic Nakajima-Zwanzig type of equation for continuously measured quantum systems.

\section{Preliminaries}
\label{sec:background}
\textbf{Notation.} $X^{\top}$ denotes the transpose of a matrix $X$, $X^{\dag}$ denotes the adjoint of a Hilbert space operator $X$ (the conjugate transpose when $X$ is a matrix). $I_n$ will denote an $n \times n$ identity matrix and $I$ can denote either an identity matrix (whose dimension can be inferred from the context), an identity map or an identity operator. $\mathrm{Tr}$ denotes the trace of a matrix or an operator. A signal (function of time) will be denoted by $V_{\cdot}$ where the subscript $\cdot$ is a placeholder for time. If a signal is clear from its context then it will be denoted simply as $V$ (without the subscript). For a signal $Y_{\cdot}$, $Y_{0:t}=\{Y_{\tau}\}_{0 \leq \tau \leq t}$. If $\mathfrak{h}_1$ and $\mathfrak{h}_2$ are Hilbert spaces, $\mathscr{L}(\mathfrak{h}_1;\mathfrak{h}_2)$ denotes the class of all linear operators mapping from $\mathfrak{h}_1$ to $\mathfrak{h}_2$. If $\mathfrak{h}_1 =\mathfrak{h} =\mathfrak{h}_2$ then it is written simply as  $\mathscr{L}(\mathfrak{h})$. If $X,Y \in \mathscr{L}(\mathfrak{h})$ then $[X,Y]=XY-YX$ and $\{X,Y\} = XY + YX$.  If $X$ is an operator on the composite Hilbert space $\mathfrak{h}_1 \otimes \mathfrak{h}_2$ then $\mathrm{Tr}_{\mathfrak{h}_j}(X)$ denotes the partial trace of $X$ by tracing out over the Hilbert space $\mathfrak{h}_j$ ($j=1,2$). If $O_j \in \mathscr{L}(\mathfrak{h}_j)$ then $O_j$ is also used as a shorthand  for the ampliation of $O_j$ to the composite Hilbert space $\mathfrak{h}_1 \otimes \mathfrak{h}_2$. Also, $\delta_{jk}$ is  the Kronecker delta and $\mathbb{E}[\cdot]$ denotes the classical expectation operator.

The set up of \cite{Nurd23} will now be revisited. Without  loss of generality, the multiple auxiliaries in \cite{Nurd23} will be combined into  a single auxiliary on a Hilbert space $\mathfrak{h}_{\rm a}$ that is finite dimensional with  $\mathrm{dim}(\mathfrak{h}_{\rm a})=n_{\rm a}$, and let $\mathfrak{h}_{\rm sa}=\mathfrak{h}_{\rm s} \otimes \mathfrak{h}_{\rm a}$ be the composite Hilbert space of the principal ($\mathfrak{h}_{\rm s}$) and auxiliary. The principal and auxiliary are coupled through $K \geq 1$ external quantum white noise fields, taken to be in the vacuum state, through the (generally time-dependent) coupling (or jump) operators $L_k(t)$, $k=1,\ldots,K$, where $L_k(t)$ is the coupling operator at time $t$ to the $k$-th quantum white noise. The coupling operators take the general form
$L_k(t) = L_{k,{\rm s}}(t) + L_{k,{\rm sa}}(t)  + L_{k,{\rm a}}(t)$, where $L_{k,{\rm s}}(t)$ is the ampliation of an operator that acts only on the principal system,  $L_{k,{\rm sa}}(t)$ acts on the principal and auxiliary, and $L_{k,{\rm a}}(t)$ is the ampliation of an operator that acts only on the auxiliary. Similarly the principal and auxiliary can also couple through a Hamiltonian $H(t)$ that is of the form  
$H(t) = H_{\rm s}(t) + H_{\rm sa}(t)  + H_{\rm a}(t)$, where, as with the coupling operator, $H_{\rm s}(t)$ is the ampliation of a Hamiltonian  that acts only on the principal system,  $H_{\rm sa}(t)$ is a Hamiltonian that acts on the principal and auxiliary, and $H_{\rm a}(t)$ is the ampliation of a Hamiltonian that acts only on the auxiliary. 

The principal system is measured by coupling it to a probe quantum white noise field, that is indicated by the index 0. The coupling to the probe is via a coupling operator $L_0(t) \in \mathscr{L}(\mathfrak{h}_{\rm s})$ for all $t$. This paper will consider the case where the probe is continuously measured via homodyne detection of the probe amplitude quadrature. However, the derivations and results herein can be straightforwardly modified for the case of continuous measurement of the phase quadrature of the probe and continuous photon counting measurement \cite{BvHJ07}. 

Let $\varrho_{{\rm sa},\cdot}$ be the conditional joint density operator of the principal and auxiliary. Then under continuous measurement of the amplitude quadrature the measurement signal $Y^Q_{\cdot}$ satisfies:
$$
dY^Q_t = \mathrm{Tr}(\varrho_{{\rm sa},t} (L_0(t)+L_0(t)^{\dag})) dt + dI_t,\; Y^Q_0=0,
$$
where $I_{\cdot}$ is a standard Wiener process, the so-called measurement shot noise.  The evolution of $\varrho_{{\rm sa},\cdot}$ is given by the operator-valued stochastic differential equation (SDE):
\begin{align}
d\varrho_{{\rm sa},t} &= \left(-i[H(t),\varrho_{\rm{sa},t}]+ \sum_{k=0}^K \mathcal{D}_{k}(t) \varrho_{{\rm sa},t} \right)dt + \left(\vphantom{L_0^{\dag}} L_0(t) \varrho_{{\rm sa},t} \right. \notag \\
&\quad + \left. \varrho_{{\rm sa},t} L_0^{\dag}(t)- \varrho_{{\rm sa},t} \mathrm{Tr}((L_0(t) + L_0^{\dag}(t)) \varrho_{{\rm sa},t})\right) dI_t, \label{eq:sa-SME}
\end{align}
where
$$
\mathcal{D}_{k}(t) \varrho_{{\rm sa},t} = L_k(t) \varrho_{{\rm sa},t} L_k(t)^{\dag} -\frac{1}{2}\left\{ L_k(t)^{\dag}L_k(t),\varrho_{{\rm sa},t} \right\} 
$$
Conditions for existence and uniqueness of a solution to \eqref{eq:sa-SME} can be found in, e.g., \cite[\S 5.1]{BG09}.  In particular, under \cite[Assumption 5.1]{BG09} \eqref{eq:sa-SME}  admits a pathwise unique continuous solution with $\varrho_{{\rm sa},t}$ a random density operator for each $t$, and uniqueness in law holds \cite[Theorem 5.6]{BG09}.

The unconditional density operator $\rho_{{\rm sa},\cdot} = \mathbb{E}\left[ \varrho_{{\rm sa},\cdot}\right]$ satisfies the GKSL quantum master equation:
\begin{align}
\dot{\rho}_{{\rm sa},t} &= -i[H(t),\rho_{\rm{sa},t}]+ \sum_{k=0}^K \mathcal{D}_{k}(t) \rho_{{\rm sa},t}. \label{eq:amplitude-meas}
\end{align}

Let $\{| \phi_j \rangle \}_{j=1,\ldots,n_{\rm a}}$ be an orthonormal basis for the auxiliary Hilbert space $\mathfrak{h}_{\rm a}$ and for any linear operator $X$ on $\mathfrak{h}_{\rm sa}$ define  $X^{jk} = \langle \phi_j| X |\phi_k\rangle  =(I \otimes \langle \phi_j  |) X ( I \otimes  |\phi_k \rangle)$, where $I$ is the identity operator on $\mathfrak{h}_{\rm s}$ and the middle term is a common shorthand  in the physics literature for  the last term.  Let $\varrho_{{\rm s},\cdot} = \mathrm{Tr}_{\mathfrak{h}_{\rm a}}(\varrho_{\rm sa})$ be the reduced conditional state of the principal system under continuous measurement of $Y^{Q}_{\cdot}$. Define $\varrho^{jk}_{{\rm s},\cdot} = \langle \phi_j | \varrho_{\rm sa} |\phi_k \rangle$ for $j,k=1,\ldots,n_{\rm a}$. Note that for each $j,k$, $\varrho^{jk}_{{\rm s},\cdot}$ is by definition an operator on the principal system. Following \cite{Nurd23}, $\varrho_{{\rm s},\cdot} $ can be computed from $\{\varrho^{kk}_{{\rm s},\cdot} \}_{k=1,\ldots,n_{\rm a}}$ as $\varrho_{{\rm s},\cdot}  = \sum_{k=1}^{n_{\rm a}} \varrho^{kk}_{{\rm s},\cdot}$. It was shown that $\varrho^{jk}_{{\rm s},\cdot}$ for $j,k=1,\ldots,n_{\rm a}$ satisfy a system of coupled SDEs that involves only principal system operators. Similarly, defining $\rho^{jk}_{{\rm s},\cdot}$ to be $\rho^{jk}_{{\rm s},\cdot} = \langle \phi_j | \rho_{\rm sa} |\phi_k \rangle$  then $\{\rho^{jk}_{{\rm s},\cdot}\}_{j,k=1,\ldots,n_{\rm a}}$ satisfy a system of coupled ODEs with the reduced unconditional state for the principal system $\rho_{{\rm s},\cdot}= \mathrm{Tr}_{\mathfrak{h}_{\rm a}}(\rho_{\rm sa})$ given by  $\rho_{{\rm s},\cdot}  = \sum_{k=1}^{n_{\rm a}} \rho^{kk}_{{\rm s},\cdot}$. 

\section{Main results}
\label{sec:results}

Let $\mathcal{P}:\mathscr{L}(\mathfrak{h}_{\rm sa}) \rightarrow \mathscr{L}(\mathfrak{h}_{\rm sa})$ be a projection superoperator that maps a density operator $\rho$ on $\mathfrak{h}_{\rm sa}$ to another density operator  $\rho'$ on $\mathfrak{h}_{\rm sa} $ ($\mathcal{P}\rho=\rho'$) such that $\mathcal{P}^2=\mathcal{P}$. It is required to satisfy the properties:
\begin{enumerate}
\item $\mathrm{Tr}_{\mathfrak{h}_{\rm a}}(\mathcal{P} \rho) = \mathrm{Tr}_{\mathfrak{h}_{\rm a}}(\rho)$.

\item $\mathcal{P} ((O \otimes  I_{\mathfrak{h}_{\rm a}}) \rho) = O  \mathcal{P} \rho$ for any operators $O$ on $\mathfrak{h}_{\rm s}$ and $\rho$ on $\mathfrak{h}_{\rm sa}$. 
 \end{enumerate}

Let $\mathcal{Q}= \mathcal{I}-\mathcal{P}$, where $\mathcal{I}$ is the identity supeoperator on $\mathscr{L}(\mathfrak{h}_{\rm sa})$. For any operator $Z \in \mathscr{L}(\mathfrak{h}_{\rm sa})$ define $Z^p=\mathcal{P}Z$ and  $Z^q=\mathcal{Q}Z$. Also, for any linear superoperator $\mathcal{X}:  \mathscr{L}(\mathfrak{h}_{\rm sa}) \rightarrow  \mathscr{L}(\mathfrak{h}_{\rm sa})$, define  $\mathcal{X}^{pp} = \mathcal{P} \mathcal{X} \mathcal{P}$,  $\mathcal{X}^{pq} = \mathcal{P} \mathcal{X} \mathcal{Q}$, $\mathcal{X}^{qp} = \mathcal{Q} \mathcal{X} \mathcal{P}$, and  $\mathcal{X}^{qq} = \mathcal{Q} \mathcal{X} \mathcal{Q}$.  Let $\mathcal{L}(t)$  be the Lindblad generator,  
\begin{align*}
\mathcal{L}(t) \varrho_{\rm{sa},t}  &= -i[H(t),\varrho_{\rm{sa},t}]+ \sum_{k=0}^K \mathcal{D}_{k}(t)   \varrho_{{\rm sa},t}, 
\end{align*}
and $\mathcal{G}(t)$ be the superoperator defined by:
$$
\mathcal{G}(t) \varrho_{\rm{sa},t} =  L_0(t) \varrho_{\rm{sa},t} +  \varrho_{\rm{sa},t}  L_0(t)^{\dag}.
$$
Note that by the assumptions on $\mathcal{P}$, $\mathcal{G}(t)$ commutes with both $\mathcal{P}$ and $\mathcal{Q}$ so that $\mathcal{P}\mathcal{G}(t) \varrho_{{\rm sa},t} = \mathcal{G}(t) \mathcal{P}\varrho_{{\rm sa},t}$ and $\mathcal{Q}\mathcal{G}(t) \varrho_{{\rm sa},t} = \mathcal{G}(t) \mathcal{Q}\varrho_{{\rm sa},t}$. 
Also,  since $L_0(t) \in \mathscr{L}(\mathfrak{h}_{\rm s})$, we have that
\begin{align*}
\lefteqn{\mathrm{Tr}((L_0(t) + L_0(t)^{\dag}) \rho_{{\rm sa},t})}\\ 
&=  \mathrm{Tr}((L_0(t) + L_0(t)^{\dag}) \mathrm{Tr}_{\mathfrak{h}_{\rm a}} (\rho_{{\rm sa},t}) )\\
&= \mathrm{Tr}((L_0(t) + L_0(t)^{\dag}) \mathcal{P} \rho_{{\rm sa},t}). 
\end{align*}
From \eqref{eq:sa-SME} and these properties of $\mathcal{P}$, it follows that:  
\begin{align}
d\varrho^{p}_{{\rm sa},t}&=(\mathcal{L}(t)^{pp} \varrho^{p}_{{\rm sa},t}+ \mathcal{L}(t)^{pq}  \varrho^{q}_{{\rm sa},t})dt \notag \\
&\quad + (\mathcal{G}(t) \varrho^{p}_{{\rm sa},t} - \varrho^{p}_{{\rm sa},t} \mathrm{Tr}((L_0(t) + L_0(t)^{\dag}) \varrho^p_{{\rm sa},t}) )dI_t \label{eq:SDE-p}\\
d\varrho^{q}_{{\rm sa},t}&=(\mathcal{L}(t)^{qp} \varrho^{p}_{{\rm sa},t}+ \mathcal{L}(t)^{qq}  \varrho^{q}_{{\rm sa},t} )dt\notag \\
&\quad + (\mathcal{G}(t) \varrho^{q}_{{\rm sa},t} - \varrho^{q}_{{\rm sa},t} \mathrm{Tr}((L_0(t) + L_0(t)^{\dag}) \varrho^p_{{\rm sa},t}) )dI_t. \label{eq:SDE-q}
\end{align}

The SDE for $\varrho^{q}_{{\rm sa},\cdot}$ is linear for a given  $\varrho^{p}_{{\rm sa},\cdot}$ and can be rewritten as:
\begin{align}
d\varrho^{q}_{{\rm sa},t}&=d\mathcal{A}_{ \varrho^p_{{\rm sa},t}}(t) \varrho^{q}_{{\rm sa},t}  + \mathcal{L}(t)^{qp}  \varrho^{p}_{{\rm sa},t}dt \label{eq:SDE-q-2}
\end{align}
where $\mathcal{A}_{ \varrho^p_{{\rm sa},t}}(t)$ is a stochastic generator given by: 
\begin{eqnarray*}
d\mathcal{A}_{ \varrho^p_{{\rm sa},t}}(t)  &=&  \mathcal{L}(t)^{qq}dt  \\
&\quad&  +( \mathcal{G}(t)-  \mathrm{Tr}((L_0(t) + L_0(t)^{\dag}) \varrho^p_{{\rm sa},t})  \mathcal{I}) dI_t, 
\end{eqnarray*}
and $\mathcal{I}$ is the identity operator on $\mathscr{L}(\mathfrak{h}_{\rm sa})$ as before. Set $\mathcal{A}_{ \varrho^p_{{\rm sa},t_0}}=0$ at an initial time $t_0$. 

Note that since $\mathfrak{h}_{\rm sa}$ is finite dimensional, all linear operators in $\mathscr{L}(\mathfrak{h}_{\rm sa})$ can be represented by finite-dimensional vectors and all superoperators acting on $\mathscr{L}(\mathfrak{h}_{\rm sa})$ can be represented as matrices.  Moreover, the representations can be chosen to be real by separating the real and imaginary parts. This allows the application of results for time-varying matrix-valued linear SDEs with additive and multiplicative Wiener (more generally semimartingale) noise (in, e.g., \cite{DY08}) in the present setting by identifying  operators and superoperators with their vector and matrix representations, respectively.

Let $\Phi_t$ be a superoperator on $\mathscr{L}(\mathfrak{h}_{\rm sa})$ that is the solution to the SDE:
\begin{align}
d\Phi_{t,t_0} = d\mathcal{A}_{ \varrho^p_{{\rm sa},t}}(t) \Phi_{t,t_0},   \label{eq:SDE-Phi}
\end{align}
with the initial condition $\Phi_{t_0,t_0} =\mathcal{I}$. Under the assumption that $\mathcal{A}_{ \varrho^p_{{\rm sa},\cdot}}(\cdot)$ is a matrix-valued semimartingale, the unique solution $\Phi_{t,t_0}$ is called the stochastic exponential of $\mathcal{A}_{ \varrho^p_{{\rm sa},\cdot}}(\cdot)$, which is invertible for each $t$ \cite{DY08}. Then the following holds:

\begin{lemma}
\label{lem:sol-varrho-q}
Suppose that \eqref{eq:sa-SME} has a unique solution that is continuous w.r.t. $t$ and adapted to the filtration generated by the Wiener process $I_{\cdot}$.  The solution to \eqref{eq:SDE-q-2} is
$$
\varrho^{q}_{{\rm sa},t} = \Phi_{t,t_0} \varrho^{q}_{{\rm sa},t_0}  +  \Phi_{t,t_0} \int_{t_0}^t \Phi_{t',t_0}^{-1} \mathcal{L}(t')^{qp}  \varrho^{p}_{{\rm sa},t'}dt'
$$
\end{lemma}
\begin{proof}
The result follows from \cite[Theorem 1.2]{DY08} by making the identification $H(t)=\varrho^{q}_{{\rm sa},t_0} + \int_{t_0}^t \mathcal{L}(t')^{qp}  \varrho^{p}_{{\rm sa},t'}dt'$ and $L(t)=\mathcal{A}_{ \varrho^p_{{\rm sa},t}}(t)$ (in this proof $H(t)$ and $L(t)$ refer to the processes as defined in \cite[Theorem 1.2]{DY08}). Since $H(t)$ and $L(t)$ are continuous processes, following \cite{DY08} let $H^c(t)$ and $L^c(t)$ be their continuous martingale part, respectively.  By the assumption of the lemma, $H$ is by definition a process with finite variation on each interval $[0,t]$, since $\mathcal{L}(t)^{qp}  \varrho^{p}_{{\rm sa},t}$ is bounded over any such interval for every sample path. Therefore, $H^c$ is a constant function and the quadratic covariation $\langle H^c(t),L^c(t) \rangle$ vanishes. Hence the term $G(t)$ in \cite[Theorem 1.2]{DY08} is simply $G(t)=H(t)$, from which the statement of the lemma follows. 
\end{proof}

\begin{theorem}
\label{thm:sol-varrho-p}
Under the assumptions of Lemma \ref{lem:sol-varrho-q}, the matrix-valued stochastic process $\varrho^{p}_{{\rm sa},\cdot}$ satisfies the SDE
\begin{align}
d\varrho^{p}_{{\rm sa},t} &=\left(\vphantom{\int_0^t}\mathcal{L}(t)^{pp} \varrho^{p}_{{\rm sa},t} \right. \notag \\
&\quad \left. + \mathcal{L}(t)^{pq} \Phi_{t,t_0} \varrho^{q}_{{\rm sa},t_0} + \int_{t_0}^t \mathcal{K}_{\rm s}(t,t')  \varrho^{p}_{{\rm sa},t'}dt' \right)dt \notag \\
&\quad + (\mathcal{G}(t) \varrho^{p}_{{\rm sa},t} - \varrho^{p}_{{\rm sa},t} \mathrm{Tr}((L_0(t) + L_0(t)^{\dag}) \varrho^p_{{\rm sa},t}) )dI_t \label{eq:reduced-SDE-p}
\end{align}
where $\Phi_{t,t_0}$ is the stochastic exponential solving the SDE \eqref{eq:SDE-Phi} and $\mathcal{K}_{\rm s}(t,t') = \mathcal{L}(t)^{pq} \Phi_{t,t_0} \Phi_{t',t_0}^{-1} \mathcal{L}(t')^{qp}$ is a two-time stochastic kernel. The conditional state $\varrho_{{\rm s},\cdot}$ is given by $\varrho_{{\rm s},\cdot} = \mathrm{Tr}_{\mathfrak{h}_{\rm a}}  (\varrho^p_{{\rm sa},\cdot})$ and the unconditional state by $\rho_{{\rm s},\cdot} = \mathbb{E}\left[\varrho_{{\rm s},\cdot}\right]$. 
\end{theorem}
\begin{proof}
Substituting the solution $\varrho^{q}_{{\rm sa},t}$ from Lemma \ref{lem:sol-varrho-q} to the right hand side of \eqref{eq:SDE-p} yields  the  SDE \eqref{eq:reduced-SDE-p}.  
\end{proof}

The intuition behind the theorem is as follows. Since the continuous-measurement is performed only on the principal  by coupling it to a measurement probe via a dissipation operator $L_0(\cdot)$ that is  a principal system operator, the stochastic measurement back-action term  in the $p$ component is decoupled from its $q$ component, as can be seen from \eqref{eq:SDE-p}. On the other hand, the same structure allows the effect of the $p$ component on the back-action term of the $q$ component to be isolated to the scalar term $\mathrm{Tr}((L_0(\cdot)+L_0(\cdot)^{\dag}) \varrho^p_{\mathrm{sa},\cdot})$. This enables the SDE for the $q$ component to be expressed as a linear SDE that is parameterized by $\varrho^p_{\mathrm{sa},\cdot}$ as given in \eqref{eq:SDE-q-2}. The linear SDE can then be solved by a stochastic version of the variation of constants formula \cite[Theorem 1.2]{DY08}.
 
A similar procedure can be followed to obtain a matrix-valued ODE for $\rho^{p}_{{\rm sa},\cdot}$ but this is a well-known procedure that leads to the so-called Nakajima-Zwanzig non-Markovian master equatiom  but here it is slightly generalized where the Liouvillian superoperator $[H(t),\cdot]$ in the standard Nakajima-Zwanzig formulation, see \cite[\S 3.2]{Bacchini11}, is replaced by the Lindbladian superoperator $\mathcal{L}(t)$. For the sake of completeness,  the Nakajima-Zwanzig equation for $\rho^{p}_{{\rm sa},\cdot}$  is given below,
\begin{align}
\dot{\rho}^{p}_{{\rm sa},t} &=  \mathcal{L}(t)^{pp} \rho^{p}_{{\rm sa},t} + \mathcal{L}(t)^{pq} \Gamma_{t,t_0}  \rho^{q}_{{\rm sa},t_0} \notag \\
&\quad + \int_{t_0}^t \mathcal{L}(t)^{pq} \Gamma_{t,t'} \mathcal{L}(t')^{qp} \rho^{p}_{{\rm sa},t'} dt' \label{eq:NZ}
\end{align}
where $\Gamma_{t,t_0}$ is the time-ordered exponential,
$\Gamma_{t,t_0} = \overleftarrow{T} \exp\left( \int_{t_0}^t \mathcal{L}^q(t') dt' \right)$ 
and  $\overleftarrow{T}$ is the chronological time ordering operator.  From the solution of the Nakajima-Zwanzig equation we obtain that $\rho_{{\rm s},\cdot} = \mathrm{Tr}_{\mathfrak{h}_{\rm a}}(\rho^p_{{\rm sa},\cdot})$. 

Note that the SDE (7) has a dependence on the initial condition $\varrho^{q}_{{\rm sa},t_0}$ similar to the Nakajima-Zwanzig equation \eqref{eq:NZ}. Under the standard initial product state assumption,  $\varrho_{{\rm sa},t_0} = \rho_{\rm s} \otimes \rho_{\rm a}$, $\mathcal{P}$ can be chosen such that $\varrho^q_{{\rm sa},t_0}=0$, e.g., $\mathcal{P} \rho=\mathrm{Tr}_{\mathfrak{h}_\mathrm{a}}(\rho)\otimes \rho_{\rm a}$. More generally, $\varrho^q_{{\rm sa},t_0}=0$ can hold under  a weaker condition than a product state (see the discussion after Theorem \ref{thm:SDE-NM-s}). Another way to remove the dependence on $\varrho^{q}_{{\rm sa},t_0}$ is if $\Phi_{t,t_0}$ is asymptotically stable in the sense that $\mathop{\lim}_{ t_0 \rightarrow -\infty} \Phi_{t,t_0}=0$ and decays sufficiently fast so that $\mathop{\lim}_{t_0 \rightarrow -\infty} \int_{t_0}^t \mathcal{K}_{\rm s}(t,t')  \varrho^{p}_{{\rm sa},t'}dt'$ exists for any $t$ and any $\{\varrho^{p}_{{\rm sa},t'},\;-\infty<t'<t \}$. Conditions for the asymptotic stability are beyond the scope of this work and have been studied elsewhere; see, e.g.,  \cite{Palamarchuk24} and the references therein. Then for $t_0 \rightarrow -\infty$ the following SDE for $\varrho^{p}_{{\rm sa},\cdot}$ is obtained that does not require knowing $\varrho^{q}_{{\rm sa},t_0}$,
\begin{align}
d\varrho^{p}_{{\rm sa},t} &=\left(\vphantom{\int_0^t}\mathcal{L}(t)^{pp} \varrho^{p}_{{\rm sa},t}  + \int_{-\infty}^t \mathcal{K}_{\rm s}(t,t') \varrho^{p}_{{\rm sa},t'}dt'  \right)dt \notag \\
&\quad + (\mathcal{G}(t) \varrho^{p}_{{\rm sa},t} - \varrho^{p}_{{\rm sa},t} \mathrm{Tr}((L_0(t) + L_0(t)^{\dag}) \varrho^p_{{\rm sa},t}) )dI_t. \label{eq:SDE-p-asymptotic}
\end{align}

The SDE presented in Theorem \ref{thm:sol-varrho-p} can be viewed as a stochastic version of the Nakajima-Zwanzig master equation for a non-Markovian quantum system under continuous measurement. As with the stochastic master equation and SSE for Markovian quantum systems, the unconditioned state  $\rho_{\mathrm{s},\cdot}$  can be computed through the relation $\rho_{\mathrm{s},\cdot} = \mathbb{E}[\mathrm{Tr}_{\mathfrak{h}_{\rm a}}(\varrho^{p}_{\mathrm{sa},\cdot})]$. However, unlike the SSE and NMQSD, which involve a state vector rather than a density matrix, the conditional matrix $\varrho^{p}_{\mathrm{sa},t}$ is of the same dimension as the unconditional matrix $\rho^{p}_{\mathrm{sa},t}$. Thus Monte Carlo simulation of $\varrho^{p}_{\mathrm{sa},\cdot}$ will not in general be a more computationally efficient method for computing the unconditional state $\rho_{\mathrm{s},\cdot}$ compared to directly solving the Nakajima-Zwanzig equation. On the other hand since the projection $\varrho^{p}_{\mathrm{sa},\cdot}$  lies on a subspace of a smaller dimension than $\varrho_{\mathrm{sa},\cdot}$,  \eqref{eq:reduced-SDE-p}  will be of interest for continuous-time quantum filtering and feedback control of non-Markovian quantum systems. In addition, unlike the pure state NMQSD, which can be used as a device for computing the unconditional reduced state $\rho_{\mathrm{s},t}$ at any time $t$ but is not in general associated with a continuous measurement process \cite{KS12,GDA25}, \eqref{eq:reduced-SDE-p}  describes the physical evolution of a non-Markovian system when it is continuously measured.

Recall the notation $X^{jk}$ introduced in \S \ref{sec:background}. The decomposition used in \cite{Nurd23} corresponds to the projection superoperator 
$\mathcal{P} X = \sum_{j=1}^{n_a} X^{jj} \otimes   | \phi_j \rangle  \langle \phi_j |$. It is directly verified  that $\mathrm{Tr}_{\mathfrak{h}_{\rm a}}(\mathcal{P}X)=\sum_{j=1}^{n_a} X^{jj} =\mathrm{Tr}_{\mathfrak{h}_{\rm a}}(X)$ and for any system operator $O$,
\begin{align*}
\mathcal{P}((O \otimes I_{\mathfrak{h}_{\rm a}})\rho)&= \sum_{j=1}^{n_a} (O\rho)^{jj} \otimes   | \phi_j \rangle  \langle \phi_j |\\
&= O \sum_{j=1}^{n_a} \rho^{jj} \otimes   | \phi_j \rangle  \langle \phi_j | =O  \mathcal{P}\rho,
\end{align*}
where the ampliation of $O$ is implied as appropriate per the notation used. However, in the approach of \cite{Nurd23} the calculations can be made explicit. Instead of evaluating the superoperators $\mathcal{L}(t)^{jk}$ for $j,k \in \{q,p\}$, $\Phi_{t,t_0}$, etc, one can work directly with block elements of the joint system and auxiliary density operators. Define $\varrho^{jk}_{{\rm s},\cdot} = \varrho^{jk}_{{\rm sa},\cdot}$ for all $j,k=1,\ldots,n_a$, where $\varrho^{jk}_{{\rm s},\cdot}$ is a principal system operator for each $j,k$.
Then $\varrho^{jk}_{{\rm s},\cdot}$ satisfies the coupled system of SDEs \cite{Nurd23}:
\begin{align}
d\varrho^{jk}_{{\rm s},t}  &=\left( i \sum_{l=1}^{n_a} (\varrho_{{\rm s},t}^{jl} H^{lk}(t) -H^{jl}(t) \varrho_{{\rm s},t}^{lk} ) \right. \notag \\
&\quad  + \sum_{m=0}^{K} \left( \sum_{r,s=1}^{n_a} L^{jr}_m(t)  \varrho^{rs}_{{\rm s},t} (L^{ks}_m(t))^{\dag} \right.  \notag\\
&\quad -\frac{1}{2}\sum_{r=1}^{n_a} \biggl((L_m^{\dag}(t) L_m(t))^{jr} \varrho^{rk}_{{\rm s},t}  \notag\\
&\qquad \left. \left. + \varrho^{jr}_{{\rm s},t} (L_m^{\dag}(t) L_m(t))^{rk}\biggr)\vphantom{\sum_k}\right) \right)dt \notag\\
&\qquad + \left( \vphantom{\sum_{l=1}^{n_a} }  L_0(t) \varrho^{jk}_{\mathrm{s},t}  + \varrho^{jk}_{\mathrm{s},t} L_0(t)^{\dag} \right. \notag\\
&\qquad \left. - \varrho^{jk}_{\mathrm{s},t} \sum_{l=1}^{n_a} \mathrm{Tr}\left((L_0(t)+L_0(t)^{\dag})\varrho^{ll}_{{\rm s},t}\right)\right) dI_t. \label{eq:coupled-SME}
\end{align}
In the above, $H^{jk}(t) = H_{\rm s}(t)\delta_{jk} +  H^{jk}_{\rm sa}(t)  + I\langle \phi_j| H_{\rm a}(t)| \phi_k \rangle$, $L_m^{jk}(t) = L_{m,{\rm s}}(t) \delta_{jk}  +  L^{jk}_{m,{\rm sa}}(t)+ I \langle \phi_j | L_{m,{\rm a}}(t) |\phi_k \rangle$, and $(L_m(t)^{\dag}L_m(t))^{jk}= \sum_{r=1}^{n_a} L_m(t)^{\dag jr} L_m(t)^{rk}$, where the $I$ in the last term for $H^{jk}(t)$ and $L_m^{jk}(t)$  is the identity matrix on the principal. On the other hand, the unconditional state $\rho^{jk}_{{\rm s},\cdot}$ satisfies the coupled system of ODEs:
\begin{align}
\dot{\rho}^{jk}_{{\rm s},t} &=  i \sum_{l=1}^{n_a} (\rho_{{\rm s},t}^{jl} H^{lk}(t) -H^{jl}(t) \rho_{{\rm s},t}^{lk} ) \notag\\
&\quad  + \sum_{m=0}^{K} \left( \sum_{r,s=1}^{n_a} L^{jr}_m(t)  \rho^{rs}_{{\rm s},t} (L^{ks}_m(t))^{\dag}  \right. \notag\\
&\qquad -\frac{1}{2}\sum_{r=1}^{n_a} \left( \vphantom{(L_m^{\dag}(t) L_m(t))^{rk}} (L_m^{\dag}(t) L_m(t))^{jr} \rho^{rk}_{{\rm s},t} \right. \notag \\
&\qquad \left.   + \rho^{jr}_{{\rm s},t} (L_m^{\dag}(t) L_m(t))^{rk})\vphantom{\sum_{k=1}}\right) \label{eq:coupled-ME}
\end{align}

Introduce the column vector of matrix-valued functions $\widetilde{\varrho}_{{\rm s},\cdot} =\left[ \begin{array}{cccc} \varrho^{11}_{{\rm s},\cdot} & \varrho^{22}_{{\rm s},\cdot} & \ldots & \varrho^{n_a n_a}_{{\rm s},\cdot} \end{array} \right]^{\dag} $ and define $\widetilde{\varrho}^0_{{\rm sc},\cdot}$ to be a vector of matrix-valued functions whose first $n_a-1$ rows are $\varrho^{12}_{{\rm s},\cdot}$, $\varrho^{13}_{{\rm s},\cdot}$, $\ldots$, $\varrho^{1 n_a}_{{\rm s},\cdot}$, followed in the next $n_a-2$ rows by $\varrho^{23}_{{\rm s},\cdot}$, $\ldots$, $\varrho^{2 n_a}_{{\rm s},\cdot}$, and so on until the last element $\rho^{(n_a-1) n_a}_{{\rm s},\cdot}$. Also, let $\widetilde{\rho}^{1}_{{\rm sc},\cdot}$ denote the vector of matrix-valued functions whose elements are adjoints of the corresponding elements of $\widetilde{\rho}^{0}_{\rm sc}$ (note that $\langle \phi_j | X | \phi_k\rangle^{\dag} = \langle \phi_k | X | \phi_j\rangle^{\dag}$ for any observable $X$). Then let 
$ \widetilde{\varrho}_{{\rm sc},\cdot} = [\begin{array}{cc} \widetilde{\varrho}_{{\rm sc},\cdot}^{0\top} & \widetilde{\varrho}_{{\rm sc},\cdot}^{1\top} \end{array}]^{\top}$. 
The coupled SDEs \eqref{eq:coupled-SME} can now be expressed for each $j,k$ with $j \neq k$ as:
\begin{align}
d\varrho^{jk}_{{\rm s},t} &=\mathcal{A}^{jk}_{00}(t)\widetilde{\varrho}_{{\rm s},t}dt + \left(\vphantom{\sum_{l=1}^{n_a}} \mathcal{A}^{jk}_{01}(t) \widetilde{\varrho}_{{\rm sc},t} dt   + dI_t \mathcal{B}^{jk}_{01,\widetilde{\varrho}_{{\rm s},t}}(t) \widetilde{\varrho}_{{\rm sc},t} \right) , \label{eq:coupled-SME-sc}
\end{align}
where $\mathcal{A}^{jk}_{00}(t)$, $j \neq k$, is a linear superoperator acting on $\widetilde{\varrho}_{{\rm s},t}$ as
\begin{align*}
\mathcal{A}^{jk}_{00}(t)\widetilde{\varrho}_{{\rm s},t}
&=  i (\varrho_{{\rm s},t}^{jj} H^{jk}(t) -H^{jk}(t) \varrho_{{\rm s},t}^{kk})  \notag \\
&\qquad  + \sum_{m=0}^{K} \left( \sum_{r=1}^{n_a} L^{jr}_m(t)  \varrho^{rr}_{{\rm s},t} (L^{kr}_m(t))^{\dag} \right.  \notag\\
&\qquad -\frac{1}{2}  \biggl( (L_m^{\dag}(t) L_m(t))^{jk} \varrho^{kk}_{{\rm s},t}  \notag\\
&\qquad \left. + \varrho^{jj}_{{\rm s},t} (L_m^{\dag}(t) L_m(t))^{jk}\biggr) \vphantom{\sum_k}\right), 
\end{align*}
$\mathcal{A}^{jk}_{01}(t)$, $j \neq k$, is a linear superoperator acting on $\widetilde{\varrho}_{{\rm sc},t}$ as
\begin{align*}
\lefteqn{\mathcal{A}^{jk}_{01}(t)\widetilde{\varrho}_{{\rm sc},t}}\\
&=  i \sum_{l=1}^{n_a} (\varrho_{{\rm s},t}^{jl} H^{lk}(t) (1-\delta_{lj}) -H^{jl}(t) \varrho_{{\rm s},t}^{lk} (1-\delta_{lk})) \notag \\
&\qquad  + \sum_{m=0}^{K} \left( \sum_{r,s=1}^{n_a} L^{jr}_m(t)  \varrho^{rs}_{{\rm s},t} (L^{ks}_m(t))^{\dag}(1-\delta_{rs}) \right.  \notag\\
&\qquad -\frac{1}{2}\sum_{r=1}^{n_a} \biggl(L_m^{\dag}(t) L_m(t))^{jr} \varrho^{rk}_{{\rm s},t} (1-\delta_{rk})  \notag\\
&\qquad \left. + \varrho^{jr}_{{\rm s},t} (L_m^{\dag}(t) L_m(t))^{rk}(1-\delta_{rj}) \biggr)  \vphantom{\sum_k}\right), 
\end{align*}
and $\mathcal{B}^{jk}_{01,\widetilde{\varrho}_{{\rm s},t}}(t)$, $j \neq k$, is a linear superoperator acting on $\widetilde{\varrho}_{{\rm sc},t}$ for each fixed $\widetilde{\varrho}_{{\rm s},t}$ as
\begin{align*}
\mathcal{B}^{jk}_{01,\widetilde{\varrho}_{{\rm s},t}}(t)\widetilde{\varrho}_{{\rm sc},t} &= L_0(t) \varrho^{jk}_{{\rm s},t} +  \varrho^{jk}_{{\rm s},t} L_0(t)^{\dag}  \notag\\
&\qquad - \varrho^{jk}_{{\rm s},t} \sum_{l=1}^{n_a} \mathrm{Tr}\left((L_0(t)+L_0(t)^{\dag})\varrho^{ll}_{{\rm s},t}\right).
\end{align*}

Similarly, for the case $j=k$, we can express 
\begin{align}
d\varrho^{jj}_{{\rm s},t} &=\mathcal{A}^{jj}_{11}(t)\widetilde{\varrho}_{{\rm sc},t} dt+ \left(\vphantom{\sum_{l=1}^{n_a}} \mathcal{A}^{jj}_{10}(t) dt  + dI_t \mathcal{B}^{jj}_{10}(t) \right)\widetilde{\varrho}_{{\rm s},t},  \notag
\end{align}
where $\mathcal{A}^{jj}_{11}(t)$ is a linear superoperator acting on $\widetilde{\varrho}_{{\rm sc},t}$ as:
\begin{align*}
\mathcal{A}^{jj}_{11}(t)\widetilde{\varrho}_{{\rm sc},t} &=  i \sum_{l=1}^{n_a} (\varrho_{{\rm s},t}^{jl} H^{lj}(t)  -H^{jl}(t) \varrho_{{\rm s},t}^{lj} )(1-\delta_{lj})  \notag \\
&\qquad  + \sum_{m=0}^{K} \left( \sum_{r,s=1}^{n_a} L^{jr}_m(t)  \varrho^{rs}_{{\rm s},t} (L^{js}_m(t))^{\dag}(1-\delta_{rs}) \right.  \notag\\
&\qquad -\frac{1}{2}\sum_{r=1}^{n_a} \left( (L_m^{\dag}(t) L_m(t))^{jr} \varrho^{rj}_{{\rm s},t} (1-\delta_{rj})\right.  \notag\\
&\qquad \left. \left. + \varrho^{jr}_{{\rm s},t} (L_m^{\dag}(t) L_m(t))^{rj} (1-\delta_{rj})\right) \vphantom{\sum_k}\right), 
\end{align*}
$\mathcal{A}^{jj}_{10}(t)$ is a linear superoperator acting on $\widetilde{\varrho}_{{\rm s},t}$ as:
\begin{align*}
\mathcal{A}^{jj}_{10}(t)\widetilde{\varrho}_{{\rm s},t} &= i (\varrho_{{\rm s},t}^{jj} H^{jj}(t) -H^{jj}(t) \varrho_{{\rm s},t}^{jj})  \notag \\
&\qquad  + \sum_{m=0}^{K} \left( \sum_{r=1}^{n_a} L^{jr}_m(t)  \varrho^{rr}_{{\rm s},t} (L^{jr}_m(t))^{\dag} \right.  \notag\\
&\qquad -\frac{1}{2}  \left( (L_m^{\dag}(t) L_m(t))^{jj} \varrho^{jj}_{{\rm s},t} \right.  \notag\\
&\qquad \left. + \varrho^{jj}_{{\rm s},t} (L_m^{\dag}(t) L_m(t))^{jj} \biggr) \vphantom{\sum_k}\right), 
\end{align*}
and $\mathcal{B}^{jj}_{10}(t)$ is a non-linear superoperator acting on $\widetilde{\varrho}_{{\rm s},t}$ as
\begin{align*}
\mathcal{B}^{jj}_{10}(t)\widetilde{\varrho}_{{\rm s},t} &= L_0(t) \varrho^{jj}_{{\rm s},t} +  \varrho^{jj}_{{\rm s},t} L_0(t)^{\dag}  \notag\\
&\qquad - \varrho^{jj}_{{\rm s},t}  \sum_{l=1}^{n_a} \mathrm{Tr}\left((L_0(t)+L_0(t)^{\dag})\varrho^{ll}_{{\rm s},t}\right).
\end{align*}

Assembling everything together, we can write
\begin{align}
d\widetilde{\varrho}_{{\rm sc},t} &= \mathcal{A}_{00}(t)\widetilde{\varrho}_{{\rm s},t}dt + (\mathcal{A}_{01}(t) dt +  dI_t \mathcal{B}_{01,\widetilde{\varrho}_{{\rm s},t}}(t))\widetilde{\varrho}_{{\rm sc},t}, \label{eq:SDE-rho-sc}
\end{align}
where $\mathcal{A}_{00}(t)$, $\mathcal{A}_{01}(t)$ and $\mathcal{B}_{01,\widetilde{\varrho}_{{\rm s},t}}(t)$  are assembled from $\mathcal{A}^{jk}_{00}(t)$, $\mathcal{A}^{jk}_{01}(t)$  and $\mathcal{B}^{jk}_{01,\widetilde{\varrho}_{{\rm s},t}}(t)$ for all $j \neq k$, respectively, in an obvious manner. Analogously, we also have
\begin{align}
d\widetilde{\varrho}_{{\rm s},t} &= \mathcal{A}_{11}(t)\widetilde{\varrho}_{{\rm sc},t} dt+ \left(\mathcal{A}_{10}(t) dt +  dI_t \mathcal{B}_{10}(t)\right)\widetilde{\varrho}_{{\rm s},t}, \label{eq:SDE-rho-s}
\end{align}
where $\mathcal{A}_{11}(t)$, $\mathcal{A}_{10}(t)$ and $\mathcal{B}_{10}(t)$  are assembled from $\mathcal{A}^{jj}_{11}(t)$, $\mathcal{A}^{jj}_{10}(t)$  and $\mathcal{B}^{jj}_{10}(t)$ for all $j$, respectively, in an obvious way.

For a fixed $\widetilde{\varrho}_{{\rm s},t}$, define $\Psi_{t,t_0}$ as the solution to the SDE,
\begin{align}
d\Psi_{t,t_0} &= \left( \mathcal{A}_{01}(t) dt   + dI_t \mathcal{B}_{01,\widetilde{\varrho}_{{\rm s},t}}(t) \right) \Psi_{t,t_0}, \label{eq:SDE-Psi}
\end{align}
with initial condition $\Psi_{t_0,t_0} = \mathcal{I}$. Here $\Psi_{t,t_0}$ is the stochastic exponential of $\int_{t_0}^t \left( \mathcal{A}_{01}(t')  dt'   + dI_{t'} \mathcal{B}_{01,\widetilde{\varrho}_{{\rm s},t'}}(t') \right)$. The next lemma and theorem follow by the same proof steps as for Lemma \ref{lem:sol-varrho-q} and Theorem   \ref{thm:sol-varrho-p} and their proofs are omitted.
\begin{lemma}
\label{lem:sol-varrho-q-2}
The solution to \eqref{eq:SDE-rho-sc} is
$$
\widetilde{\varrho}_{{\rm sc},t} = \Psi_{t,t_0} \widetilde{\varrho}_{{\rm sc},t_0}  + \Psi_{t,t_0}\int_{t_0}^t \Psi_{t',t_0}^{-1} \mathcal{A}_{00}(t')\widetilde{\varrho}_{{\rm s},t'} dt'
$$
\end{lemma}

 \begin{theorem}
\label{thm:SDE-NM-s}
The matrix-valued stochastic process $\widetilde{\varrho}_{{\rm s},\cdot}$ satisfies the SDE
\begin{align}
d\widetilde{\varrho}_{{\rm s},t} &= \left(\mathcal{A}_{11}(t) \Psi_{t,t_0} \widetilde{\varrho}_{{\rm sc},t_0}  + \int_{t_0}^t \mathcal{K}_{\rm s}(t,t') \widetilde{\varrho}_{{\rm s},t'} dt' \right.
 \notag \\
&\qquad \left. + \mathcal{A}_{10}(t)\widetilde{\varrho}_{{\rm s},t}  \vphantom{\int_0^t} \right) dt + \mathcal{B}_{10}(t) \widetilde{\varrho}_{{\rm s},t} dI_t , \label{eq:SDE-s} \end{align}
where $\Psi_{t,t_0}$ is the unique solution to the SDE \eqref{eq:SDE-Psi} and $\mathcal{K}_{\rm s}(t,t') = \mathcal{A}_{11}(t)\Psi_{t,t_0} \Psi_{t',t_0}^{-1}\mathcal{A}_{00}(t')$ is a two-time stochastic kernel. The conditional state is then given by $\varrho_{{\rm s},\cdot} = \sum_{k=1}^{n_a} \varrho^{kk}_{{\rm s},\cdot}$ (where $\varrho^{kk}_{{\rm s},\cdot}$ appears as the $k$-th element of $\widetilde{\varrho}_{{\rm s},\cdot}$) and the unconditional state by $\rho_{{\rm s},\cdot} = \sum_{k=1}^{n_a} \mathbb{E}\left[\varrho^{kk}_{{\rm s},\cdot}\right]$. 
\end{theorem}

Note that $\widetilde{\varrho}_{{\rm sc},t_0}$ vanishes for a product initial state $\varrho_{\mathrm{sa},t_0}=\rho_{\rm s}\otimes \rho_{\rm a}$ when  $\{|\phi_k \rangle\}$ are chosen as the eigenstates of $\rho_{\rm a}$ and for  initial states of the form $\varrho_{\mathrm{sa},t_0}=\frac{1}{n_a}\sum_{k=1}^{n_a} \rho_{\mathrm{s},k} \otimes |\phi_k \rangle \langle \phi_k|$. Also, if $\Psi_{t,t_0}$ is asymptotically stable in the sense that $\mathop{\lim}_{ t_0 \rightarrow -\infty} \Psi_{t,t_0}=0$ and decays sufficiently fast so that $\mathop{\lim}_{t_0 \rightarrow -\infty} \int_{t_0}^t \mathcal{K}_{\rm s}(t,t')  \widetilde{\varrho}_{{\rm s},t'}dt' $ exists for any $t$ and any $\{\widetilde{\varrho}_{{\rm s},t'},\;-\infty<t'<t \}$ then for $t_0 \rightarrow -\infty$,  one arrives at the following SDE for $\widetilde{\varrho}_{{\rm s},\cdot}$ that requires no knowledge of $\widetilde{\varrho}_{{\rm sc},t_0}$,
\begin{align}
d\widetilde{\varrho}_{{\rm s},t} &=  \left(\int_{-\infty}^t \mathcal{K}_{\rm s}(t,t')\widetilde{\varrho}_{{\rm s},t'} dt'
 + \mathcal{A}_{10}(t)\widetilde{\varrho}_{{\rm s},t} \right)dt \notag\\
&\qquad +  \mathcal{B}_{10}(t)\widetilde{\varrho}_{{\rm s},t} dI_t.  \label{eq:SDE-s-asymptotic}
\end{align}
The reader is again referred to \cite{Palamarchuk24} and the references therein for conditions for asymptotic stability.

\section{Conclusion}
\label{sec:conclu}
This paper has derived SDEs for the reduced evolution of continuously-monitored non-Markovian quantum systems, reduced meaning that the SDEs involve only  operators of the principal system of interest and not on any operators of the auxiliary. These SDEs give a more compact representation of the stochastic dynamics of the principal system with a density operator  of dimension $n_s^2 \times n_a$ as opposed to dimension $(n_s \times n_a)^2$ for the Markovian embedding, where $n_s=\mathrm{dim}(\mathfrak{h}_{\mathrm{s}})$. The dimension $n_a$ can in principle be  large, possibly much larger than $n_s$. 

The two-time  kernel functions $\mathcal{K}_{\rm s} $ given  by  $\mathcal{K}_{\rm s}(t,t') = \mathcal{L}(t)^{pq} \Phi_{t,t_0} \Phi_{t',t_0}^{-1}\mathcal{L}(t')^{qp}$ in  SDEs \eqref{eq:SDE-p} and  \eqref{eq:SDE-p-asymptotic}, or $\mathcal{K}_{\rm s}(t,t') = \mathcal{A}_{11}(t) \Psi_{t,t_0} \Psi_{t,t_0}^{-1}\mathcal{A}_{00}(t')$ in SDEs \eqref{eq:SDE-s} and \eqref{eq:SDE-s-asymptotic}, which emerge in the derivation, are anticipated to play an important role.  These kernel functions summarise the influence of the auxiliary on the principal system. Therefore, relevant questions are whether such kernel functions can be efficiently approximated in, for instance, a data-driven way from experimental continuous-measurement data, and how feedback controllers can be designed for non-Markovian quantum systems described by the SDEs derived in this work. 
\bibliographystyle{ieeetran}
\bibliography{cdc2025}

\begin{thebibliography}{10}
\providecommand{\url}[1]{#1}
\csname url@samestyle\endcsname
\providecommand{\newblock}{\relax}
\providecommand{\bibinfo}[2]{#2}
\providecommand{\BIBentrySTDinterwordspacing}{\spaceskip=0pt\relax}
\providecommand{\BIBentryALTinterwordstretchfactor}{4}
\providecommand{\BIBentryALTinterwordspacing}{\spaceskip=\fontdimen2\font plus
\BIBentryALTinterwordstretchfactor\fontdimen3\font minus
  \fontdimen4\font\relax}
\providecommand{\BIBforeignlanguage}[2]{{%
\expandafter\ifx\csname l@#1\endcsname\relax
\typeout{** WARNING: IEEEtran.bst: No hyphenation pattern has been}%
\typeout{** loaded for the language `#1'. Using the pattern for}%
\typeout{** the default language instead.}%
\else
\language=\csname l@#1\endcsname
\fi
#2}}
\providecommand{\BIBdecl}{\relax}
\BIBdecl

\bibitem{HP84}
R.~L. Hudson and K.~R. Parthasarathy, ``{Quantum Ito's formula and stochastic
  evolution},'' \emph{Commun. Math. Phys.}, vol.~93, pp. 301--323, 1984.

\bibitem{GZ04}
C.~W. Gardiner and P.~Zoller, \emph{Quantum Noise: A Handbook of Markovian and
  Non-Markovian Quantum Stochastic Methods with Applications to Quantum
  Optics}, 3rd~ed.\hskip 1em plus 0.5em minus 0.4em\relax Berlin and New York:
  Springer-Verlag, 2004.

\bibitem{BvHJ07}
L.~Bouten, R.~{van Handel}, and M.~R. James, ``An introduction to quantum
  filtering,'' \emph{SIAM J. Control Optim.}, vol.~46, pp. 2199--2241, 2007.

\bibitem{WM10}
H.~M. Wiseman and G.~J. Milburn, \emph{Quantum Measurement and Control}.\hskip
  1em plus 0.5em minus 0.4em\relax Cambridge University Press, 2010.

\bibitem{NY17}
H.~I. Nurdin and N.~Yamamoto, \emph{Linear Dynamical Quantum Systems: Analysis,
  Synthesis, and Control}, ser. Communications and Control Engineering.\hskip
  1em plus 0.5em minus 0.4em\relax Cham: Switzerland: Springer, 2017.

\bibitem{NC10}
M.~Nielsen and I.~Chuang, \emph{Quantum Computation and Quantum Information:
  10th Anniversary Edition}.\hskip 1em plus 0.5em minus 0.4em\relax Cambridge
  University Press, 2010.

\bibitem{GJN11}
J.~Gough, M.~R. James, and H.~I. Nurdin, ``Quantum master equation and filter
  for systems driven by filed in a single photon state,'' in \emph{Proc. IEEE
  Conference on Decision and Control}, 2011, pp. 5570--5576.

\bibitem{GJNC12}
J.~Gough, M.~R. James, H.~I. Nurdin, and J.~Combes, ``Quantum filtering for
  systems driven by fields in single-photon states or superposition of coherent
  states,'' \emph{Phys. Rev. A}, vol.~86, p. 043819, 2012.

\bibitem{GJN13}
J.~E. Gough, M.~R. James, and H.~I. Nurdin, ``Quantum filtering for systems
  driven by fields in single photon states and superposition of coherent states
  using non-markovian embeddings,'' \emph{Quantum Inf Process}, vol.~12, pp.
  1469--1499, 2013.

\bibitem{GJN14}
------, ``Quantum trajectories for a class of continuous matrix product input
  states,'' \emph{New J. Phys.}, vol.~16, p. 075008, 2014.

\bibitem{Clerk10}
A.~A. Clerk, M.~H. Devoret, S.~M. Girvin, F.~Marquardt, and R.~J. Schoelkopf,
  ``Introduction to quantum noise, measurement, and amplification,'' \emph{Rev.
  Mod. Phys.}, vol.~82, p. 1155, 2010.

\bibitem{Nurd23}
H.~I. Nurdin, ``Markovian embeddings of non-{M}arkovian quantum systems:
  Coupled stochastic and quantum master equations for non-{M}arkovian quantum
  systems,'' in \emph{Proc. IEEE Conference on Decision and Control, pp. 5570 -
  5576}, 2023, pp. 5570--5576.

\bibitem{Imamoglu94}
A.~Imamoglu, ``Stochastic wave-function approach to non-{M}arkovian systems,''
  \emph{Phys. Rev. A}, vol.~50, no.~5, pp. 3650--3653, 1994.

\bibitem{DBG01}
B.~J. Dalton, S.~M. Barnett, and B.~M. Garraway, ``Theory of pseudomodes in
  quantum optical processes,'' \emph{Phys. Rev. A}, vol.~64, p. 053813, 2001.

\bibitem{Mascherpa20}
F.~Mascherpa \emph{et~al.}, ``Optimized auxiliary oscillators for the
  simulation of general open quantum systems,'' \emph{Phys. Rev. A}, vol. 101,
  p. 052108, 2020.

\bibitem{CAG19}
F.~Chen, E.~Arrigoni, and M.~Galperin, ``Markovian treatment of non-{M}arkovian
  dynamics of open fermionic systems,'' \emph{New J. Phys.}, vol.~21, p.
  123035, 2019.

\bibitem{CKS17}
J.~Combes, J.~Kerckhoff, and M.~Sarovar, ``The {SLH} framework for modeling
  quantum input-output networks,'' \emph{Adv. Phys. X}, vol.~2, no. 784, 2017.

\bibitem{HC93}
H.~Carmichael, \emph{An Open Systems Approach to Quantum Optics}.\hskip 1em
  plus 0.5em minus 0.4em\relax Berlin: Springer, 1993.

\bibitem{DGS98}
L.~Di\'{o}si, N.~Gisin, and W.~T. Strunz, ``Non-{M}arkovian quantum state
  diffusion,'' \emph{Phys. Rev. A}, vol.~58, p. 1699, 1998.

\bibitem{KS12}
S.~Kr\"{o}nke and W.~T. Strunz, ``Non-{M}arkovian quantum trajectories,
  instruments and time-continuous measurements,'' \emph{J. Phys. A}, vol.~45,
  no.~5, p. 055305, 2012.

\bibitem{GDA25}
J.~Gough, H.~Ding, and N.~H. Amini, ``Reproducing kernel {H}ilbert space
  approach to non-{M}arkovian quantum stochastic models,'' \emph{J. Math.
  Phys.}, p. 042102, 2025.

\bibitem{BG09}
A.~Barchielli and M.~Gregoratti, ``The stochastic master equation: {P}art
  {II},'' in \emph{Quantum Trajectories and Measurements in Continuous Time},
  ser. Lecture Notes in Physics, 2009, vol. 782, pp. 111--123.

\bibitem{DY08}
J.~Duan and J.~Yan, ``General matrix-valued inhomogeneous linear stochastic
  differential equations and applications,'' \emph{Stat. Prob. Lett.}, vol.~78,
  no.~15, pp. 2361--2365, 2008.

\bibitem{Bacchini11}
B.~Vacchini, ``Non-{M}arkovian dynamics in open quantum systems,'' 2011,
  lecture notes 2224-2 in the ``School on New Trends in Quantum Dynamics and
  Quantum Entanglement" (14-18 February 2011), The Abdussalam Centre for
  Theoretical Physics.

\bibitem{Palamarchuk24}
E.~S. Palamarchuk, ``On asymptotic behavior of solutions of linear
  multidimensional stochastic differential equations with multiplicative
  noise,'' \emph{Theory Probab. Appl.}, vol.~69, no.~3, pp. 372--390, 2024.

\end{thebibliography}

\end{document}